\documentclass[letterpaper,twocolumn,10pt]{article}
\usepackage{usenix,epsfig,endnotes}
\usepackage{mdframed}
\usepackage{enumitem}
\usepackage{amsmath}
\usepackage{amsthm}
\usepackage{amssymb}
\usepackage{subcaption}
\usepackage{booktabs}
\usepackage{threeparttable}

\usepackage{tikz}
\usetikzlibrary{positioning,calc}
\usetikzlibrary{decorations.pathreplacing, shapes.misc, shapes.geometric}
\usepackage{xcolor}
\usepackage{algorithm}
\usepackage{algorithmicx}
\usepackage{algpseudocode}
\usepackage{tablefootnote}
\usepackage{svg}
\usepackage{epigraph}

\usepackage[most]{tcolorbox}
\usepackage{fontawesome5}

\theoremstyle{definition}
\newtheorem{theorem}{Theorem}[section]
\newtheorem{definition}{Definition}[section]
\theoremstyle{remark}
\newtheorem{lemma}[theorem]{Lemma}

\microtypecontext{spacing=nonfrench}

\newcommand{\Rin}{R_{\text{in}}}
\newcommand{\Win}{W_{\text{in}}}
\newcommand{\Rout}{R_{\text{out}}}
\newcommand{\Wout}{W_{\text{out}}}
\newcommand{\negl}{\text{negl}}

\newenvironment{intuitionbox}[1][]%
  {\begin{mdframed}[backgroundcolor=blue!12,
                    linecolor=blue,
                    roundcorner=3pt,
                    innertopmargin=0.6\baselineskip,
                    innerbottommargin=0.6\baselineskip,
                    skipabove=0.4\baselineskip, skipbelow=\baselineskip]
   \noindent\textbf{#1}\quad}%
  {\end{mdframed}}

\usepackage{comment}

\newif\ifshowresponses
\showresponsesfalse  

\ifshowresponses

  \newmdenv[
    skipabove=4pt,
    skipbelow=4pt,
    linecolor=blue!60!black,
    backgroundcolor=blue!6!white,
    linewidth=0.35pt,
    roundcorner=2pt,
    innerleftmargin=4pt,
    innerrightmargin=4pt,
    innertopmargin=4pt,
    innerbottommargin=4pt,
    font=\scriptsize,
    nobreak=true, 
  ]{response}

  \newmdenv[
    skipabove=4pt,
    skipbelow=4pt,
    linecolor=red!65!black,
    backgroundcolor=red!6!white,
    linewidth=0.35pt,
    roundcorner=2pt,
    innerleftmargin=4pt,
    innerrightmargin=4pt,
    innertopmargin=4pt,
    innerbottommargin=4pt,
    font=\scriptsize,
    nobreak=true,
  ]{responseA}

  \newmdenv[
    skipabove=4pt,
    skipbelow=4pt,
    linecolor=violet!60!black,
    backgroundcolor=violet!7!white,
    linewidth=0.35pt,
    roundcorner=2pt,
    innerleftmargin=4pt,
    innerrightmargin=4pt,
    innertopmargin=4pt,
    innerbottommargin=4pt,
    font=\scriptsize,
    nobreak=true,
  ]{responseB}

\else
  \excludecomment{response}
  \excludecomment{responseA}
  \excludecomment{responseB}
\fi

\begin{document}

\date{}

\title{\Large \bf Towards Anonymous Neural Network Inference}

\author{
{\rm Liao Peiyuan}\\
}

\maketitle

\subsection*{Abstract}

We introduce \textit{funion}, a system providing end-to-end sender-receiver unlinkability for neural network inference. By leveraging the Pigeonhole storage protocol and BACAP (blinding-and-capability) scheme from the Echomix anonymity system, \textit{funion} inherits the provable security guarantees of modern mixnets. Users can anonymously store input tensors in pseudorandom storage locations, commission compute services to process them via the neural network, and retrieve results with no traceable connection between input and output parties. This store-compute-store paradigm masks both network traffic patterns and computational workload characteristics, while quantizing execution timing into public latency buckets. Our security analysis demonstrates that \textit{funion} inherits the strong metadata privacy guarantees of Echomix under largely the same trust assumptions, while introducing acceptable overhead for production-scale workloads. Our work paves the way towards an accessible platform where users can submit fully anonymized inference queries to cloud services.

\section{Introduction} 

\epigraph{The advancement of a gradually ubiquitous technology should not come at the cost of diminishing personal privacy or increasing reliance on operator benevolence.}{}

Neural networks, especially large language models (LLM), are increasingly deployed via cloud and third-party services. This trend is driven by their emergent capabilities and immense computational and memory demands, making on-premise deployment impractical for many users \cite{bommasani2021opportunities}. Using neural networks in critical services (search engines, medical or legal assistants, enterprise automation, etc.) yields great utility, but it also raises privacy concerns. Users must send sensitive data (prompts, inputs) to remote servers, providing a need for \textit{anonymous} neural network execution. 

\begin{figure}[ht]
\centering
\begin{tikzpicture}[
  scale = 0.6,
  font  = \sffamily,
  box/.style     ={draw,rectangle,minimum width=2.0cm,minimum height=0.9cm,
                   font=\ttfamily\bfseries,align=center,thick},
  cloud/.style   ={draw,ellipse,minimum width=2.4cm,minimum height=1.15cm,
                   font=\ttfamily\bfseries,align=center,thick},
  courier/.style ={draw,rectangle,minimum width=2.0cm,minimum height=0.9cm,
                   font=\ttfamily\bfseries,align=center,thick,densely dashed},
  replica/.style ={draw,cylinder,shape border rotate=90,aspect=0.32,
                   minimum height=0.75cm,minimum width=0.6cm,
                   font=\scriptsize\ttfamily,align=center,thick},
  capbox/.style  ={draw,rectangle,minimum width=0.30cm,minimum height=0.30cm,
                   font=\ttfamily\bfseries\tiny,thick},
  arrow/.style   ={->,>=latex,thick},
  dasharrow/.style={->,>=latex,thick,dashed},
  label/.style   ={font=\scriptsize\ttfamily},
]

\node[box]   (alice)                          {ALICE};
\node[cloud] (mixnet) [below=1.40cm of alice] {ECHOMIX};

\node[courier] (bobC)    at ($(mixnet.south)+(-4.0,-2.20)$) {BOB\\\scriptsize(storage)};
\node[courier] (charlie) [below=2.20cm of mixnet]           {CHARLIE\\\scriptsize(compute)};
\node[courier] (benC)    at ($(mixnet.south)+( 4.0,-2.20)$) {BEN\\\scriptsize(storage)};

\node[replica] (bobR1) [below=0.25cm of bobC] {};
\node[replica] (bobR2) [below=0.05cm of bobR1] {};
\node[below=0.00cm of bobR2,font=\scriptsize\ttfamily] {replicas$_{\text{ctx in}}$};

\node[replica] (benR1) [below=0.25cm of benC] {};
\node[replica] (benR2) [below=0.05cm of benR1] {};
\node[below=0.00cm of benR2,font=\scriptsize\ttfamily] {replicas$_{\text{ctx out}}$};

\node[capbox] at ($(bobC.east)+(0.60, 0.30)$) {\texttt{B}};
\node[capbox] at ($(benC.east)+(0.60, 0.30)$) {\texttt{E}};

\draw[arrow] (alice) -- (mixnet);
\draw[arrow] (mixnet) to[out=240,in= 90] node[label,above]  {1} (bobC.north);
\draw[arrow] (mixnet) to[out=240,in= 90] node[label,left] {2} (charlie);
\draw[arrow] (mixnet) to[out=300,in= 90] node[label,above] {5} (benC.north);

\draw[dasharrow] (charlie) to[out=90,in=300] node[label,right] {3,4} (mixnet);
\draw[dasharrow] (mixnet) to[out=150,in= 90] node[label,left]  {3} (bobC.north);
\draw[dasharrow] (mixnet) to[out=390,in= 90] node[label,right] {4} (benC.north);

\end{tikzpicture}

\vspace{0.4\baselineskip}
\scriptsize
\textbf{1} Upload via $\Win$\quad
\textbf{2} Dispatch $\Rin\!\|\!\Wout$\quad
\textbf{3} Fetch with $\Rin$\\
\textbf{4} Store with $\Wout$\quad
\textbf{5} Fetch with $\Rout$ \\
solid: client traffic \quad dashed: service traffic
\vspace{-0.55\baselineskip}

\caption{\textit{funion} store $\rightarrow$ compute $\rightarrow$ store workflow.  
Bob and Ben are \emph{storage couriers} inside the mixnet; Charlie is a \emph{compute courier} whose fetch/store requests are themselves anonymized by first entering the mixnet. \texttt{B}/\texttt{E} mark the BACAP boxes handled along each chain.}
\label{fig:funion-architecture}
\end{figure}
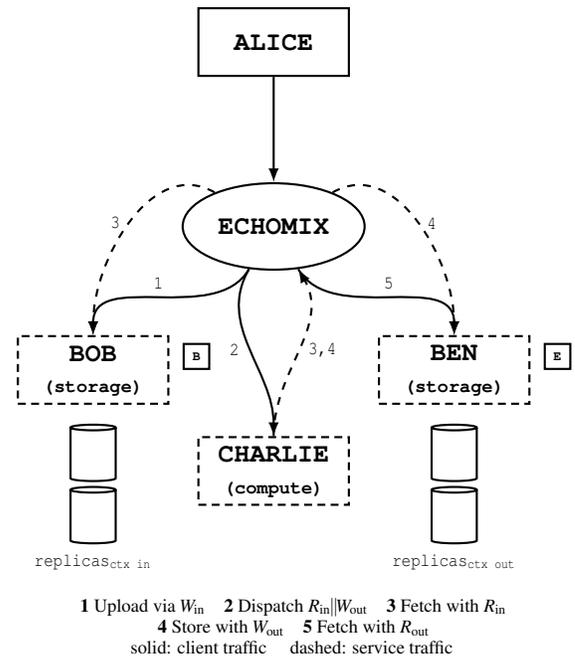

When a user queries a model, neither the content of the query, the response, nor the fact that such a user made the query should be exposed to prying eyes. At the time of writing, most cloud LLM services lack any type of anonymity; their operating model requires the ability to identify, meter, and moderate neural network executions to charge their users, maintain compliance, and train on user data \cite{bender2021dangers}.

Fortunately, the principles of decades-old anonymous communication networks can be adapted for neural network inference. Specifically, we leverage Echomix, a state-of-the-art mix network \cite{infeld2025echomix}, to provide anonymity guarantees for neural network inference. This approach, which we call \textit{funion}, orchestrates the inference of neural networks across service nodes, using the mix network's anonymity properties to ensure that no single server can link an input to its final output, even if the model is public.

At a high level, \textit{funion} separates the concerns of anonymity and neural (network) inference. The mix network handles anonymity through layered routing, cover traffic, and cryptographic protections, while the service nodes focus on efficient distributed neural inference. They store input tensors, perform computations on them, and allow clients to retrieve results-all while maintaining the anonymity guarantees provided by the underlying mixnet.

Our approach builds upon proven security mechanisms in Echomix-specifically the BACAP (Blinding-And-Capability) scheme and Pigeonhole storage protocol-to achieve sender-receiver unlinkability even against sophisticated adversaries.

\noindent Our contributions include:
\begin{enumerate}[leftmargin=*]
    \item A system design that leverages Echomix anonymity mechanisms for neural inference through a store-compute-store workflow
    \item A formal security analysis that shows correspondence to the established properties of BACAP and Echomix, thereby inheriting its anonymity guarantees
    \item A theoretical estimation showing the latency and bandwidth overhead compared to an unprotected inference server on a particular workload serving Llama-3-70B
\end{enumerate}

\section{Threat Model}

To establish our anonymity guarantees, we adopt and extend the threat model from Echomix \cite{infeld2025echomix}. 

\subsection{Client's Perspective}

Neural inference services operate fundamentally as remote procedure calls. When a client, Alice, wishes to utilize a neural network model, she provides input tensors $x \in \mathcal{X}$ and receives output tensors $y \in \mathcal{Y}$. In typical inference systems, Alice connects directly to a provider's API endpoint, authenticates, sends serialized tensors, and awaits results. This workflow, while computationally efficient, exposes significant metadata.

\subsection{What Metadata Might We Leak}

There is a variety of metadata we may leak even if we are just querying an API:

\begin{itemize}[leftmargin=*]
    \item \textbf{Identity information}: Client identifiers, authentication tokens, and network addresses that directly link a user to specific queries
    
    \item \textbf{Temporal patterns}: Timing of requests that may correlate with external events or reveal usage frequency patterns characteristic to specific users
    
    \item \textbf{Workload characteristics}: Input tensor dimensions, sequence lengths, and computational requirements that can serve as request fingerprints
    
    \item \textbf{Input-output linkability}: Correlation between which inputs produce which outputs, potentially revealing sensitive query patterns
    
    \item \textbf{Request volume}: Number and size of queries that may indicate business activities or usage patterns
\end{itemize}

Even with encrypted content over SSL/TLS, this metadata leakage creates significant privacy risks, particularly for users whose inference queries contain sensitive information or would reveal privileged activities \cite{carlini2024remote}. Systems that protect only content but neglect metadata protection leave substantial attack surfaces for both sophisticated adversaries with passive monitoring and active manipulation capabilities, and compliant but curious providers who can derive value from the data they observe.

\subsection{Adversary Classes}

We now characterize the adversaries who might attempt to exfiltrate it. Following Echomix's threat model, we consider multiple adversary classes:

\begin{enumerate}[leftmargin=*]
    \item \textbf{Global Passive Adversary (GPA)}: The GPA can observe all network traffic between nodes but cannot modify messages or compromise nodes. This adversary sees the communication metadata: timing, size, and routing information of all packets.
    
    \item \textbf{Partially Global Active Adversary}: Beyond GPA capabilities, this adversary controls a subset of nodes $\mathcal{N}_{\mathcal{A}} \subset \mathcal{N}$. Controlled nodes can deviate from the protocol, drop messages, or perform arbitrary computations.
    
    \item \textbf{Honest-but-Curious Service}: This adversary controls a service node, observing all hidden states it processes but following the protocol specification.
\end{enumerate}

Each adversary class is further characterized by its traits:

\begin{itemize}[leftmargin=*]
    \item \textbf{Model Knowledge}: All adversaries have complete knowledge of the public model parameters $\theta$, as well as how to run the neural network (attention mechanisms, activation functions, etc.).
    
    \item \textbf{Hidden State Access}: For a compromised or a curious service $s \in \mathcal{S}$ processing layer $j$, the adversary observes the hidden state $h_j$ produced during the forward pass of the neural inference. We acknowledge that services process plaintext hidden states, but protection of these states is outside the scope of this work.
    
    \item \textbf{Has Context}: The adversary may have access to rich contextual information from other data sources and can supplement network observations with this context.
    
    \item \textbf{Sophistication}: The adversary may have large computational resources and can perform cryptanalysis on par with frontier research, including access to quantum computing resources.
\end{itemize}

\subsection{Attack Vectors}
\label{sec:attack-vectors}

Given these adversary capabilities, we need to protect against three fundamental categories of attacks that threaten mix network anonymity \cite{infeld2024mixnet} in the context of neural network inference:

\begin{enumerate}[leftmargin=*]
    \item \textbf{Traffic Pattern Analysis} \cite{troncoso2009bayesian, danezis2005statistical}: 
    Attacks that exploit statistical correlations in observable network metadata over time. These attacks analyze message volumes, frequencies, and recipient distributions to gradually narrow the anonymity set through correlation and intersection. In neural inference, an adversary might identify common input size patterns (e.g., prompt lengths) or output generation characteristics specific to certain users. Even with encrypted payloads, these analyses can eventually distinguish actual communication relationships from cover traffic by observing multiple rounds of communication.
    
    \item \textbf{Active Manipulation} \cite{serjantov2003flood, danezis2008bridging, danezis2005compulsion, helsingius1996penet}: 
    Attacks where the adversary directly interferes with network operation through selective blocking, infrastructure compromise, node infiltration, or cryptographic subversion. For neural inference, this could involve infiltrating a service node, delaying specific tensor transfers, or manipulating routing to force predictable packet paths. The essence of these attacks is forcing the mix network into states where anonymity guarantees break down, creating observable discontinuities that reveal message paths or content.
    
    \item \textbf{Timing Side Channels} \cite{carlini2024remote, feigenbaum2010timing, shmatikov2006timing}: 
    Attacks that specifically exploit fine-grained timing measurements and processing duration correlations. Unlike broad pattern analysis, timing side channels directly measure computational or transmission latencies to infer properties about the workload. For neural inference, this is particularly dangerous as computation time varies significantly with input complexity (e.g., sequence length, prompt types). An adversary measuring the precise time between input submission and result retrieval can potentially determine characteristics of the query or even identify the specific input among a set of candidates, even when the packet contents are encrypted and traffic patterns are obfuscated.
\end{enumerate}

\subsection{Security Objectives}

Assuming the query payload is already encrypted in transit (e.g., via TLS), these attack vectors determine the security properties our system must achieve:

\begin{enumerate}[leftmargin=*]
    \item \textbf{Sender-Receiver Third-Party Unlinkability (SRTU)}: For any two honest senders $S_1, S_2 \in \mathcal{U}$ and receivers $R_1, R_2 \in \mathcal{U}$, the adversary cannot distinguish between the communication patterns $\{S_1 \rightarrow R_1, S_2 \rightarrow R_2\}$ and $\{S_1 \rightarrow R_2, S_2 \rightarrow R_1\}$ with non-negligible advantage. Formally, let $X$ be the event that senders and receivers are paired as in the first pattern, and $X'$ that they are paired as in the second pattern, with $X_\mathcal{A}$ being the adversary's guess. Then:
    \begin{equation*}
        |\Pr(X_\mathcal{A}|X) - \Pr(X_\mathcal{A}|X')| \leq \delta
    \end{equation*}
    for negligible $\delta$.

    \begin{intuitionbox}
    In a neural inference setting, the "sender" is the client who uploads the
    input tensor(s), and the "receiver" is that \emph{same} client when she later
    fetches the result tensor(s).  
    SRTU therefore guarantees that no observer-not even a global passive
    adversary-can link a particular uploaded input to the matching fetched
    output.
    \end{intuitionbox}
    
    \item \textbf{Input-Output Unlinkability (IO-U)}: For any two input tensors and their corresponding output tensors, the adversary cannot determine with non-negligible advantage which input produced which output. This property ensures that even the storage boxes containing inputs and outputs cannot be linked based on their identifiers or timing patterns.
    
    \begin{intuitionbox}
    IO-U strictly implies SRTU in the \textit{funion} workflow; see \S\ref{sec:security} for proof.
    \end{intuitionbox}
    
    \item \textbf{Sender Online Unobservability}: For any sender $S \in \mathcal{U}$, the adversary cannot determine with non-negligible advantage whether $S$ is communicating with any receiver ($\{S \rightarrow\}$) or not ($\{S \not\rightarrow\}$). This is achieved through the unobservable coupling of application traffic with echo decoy traffic.
    
    \item \textbf{Receiver Unobservability}: For any receiver $R \in \mathcal{U}$, the adversary cannot determine with non-negligible advantage whether any sender is communicating with $R$ ($\{\rightarrow R\}$) or not ($\{\not\rightarrow R\}$). This is achieved through the Sphinx packet format's Single Use Reply Blocks (SURBs) \cite{danezis2009sphinx} and the carefully designed echo-based traffic patterns.

    \item \textbf{Protection Against Traffic Analysis}: Through memoryless mixing, \textit{funion} ensures that network observers cannot correlate input and output packets. Each message has its delay independently sampled from the exponential distribution:
    \begin{equation*}
        f(x \geq 0, \lambda) = \lambda e^{-\lambda x}
    \end{equation*}
    which has the critical memoryless property that $\Pr(X > t + s | X > t) = \Pr(X > s)$ for any $t, s \geq 0$. This means that at any point in time, each message sitting in a mix node has the same probability distribution of remaining delay, regardless of how long it has already been waiting.
    
    \item \textbf{Computation-Time Unlinkability \cite{carlini2024remote}}: Correlation between a user's input-dependent compute cost and any observable on-wire timing should be negligible.
\end{enumerate}

\subsection{Security Non-Objectives}

Certain adversarial behaviors fall outside our current threat model, but may be addressed in future works:

\begin{enumerate}[leftmargin=*]
    \item \textbf{Hidden State Privacy}: Our design intentionally does not protect the plaintext hidden activations processed by service nodes. While this is an important privacy concern for neural inference, addressing it would require techniques like cryptographic inference, obfuscated representations, or trusted execution environments, which we defer to future work.
    
    \item \textbf{Dishonest Computation Nodes}: Services that claim to perform specific computations but deliberately provide incorrect results or follow a different algorithm are not addressed in our security framework. The system assumes computational correctness from non-compromised nodes.
    
    \item \textbf{Straggler Nodes}: Nodes that introduce arbitrary delays or fail to complete assigned computations within reasonable time bounds fall outside our threat scope. Our exponential delay distribution provides protection against timing analysis but assumes eventual computation completion.
    
    \item \textbf{Resource Exhaustion Attacks}: While \textit{funion} includes mechanisms to handle standard network congestion, deliberate resource exhaustion attacks aimed at degrading service quality rather than breaking anonymity are not directly addressed.
\end{enumerate}

\subsection{Pigeonhole Storage, BACAP, and SURB}
\label{sec:Pigeonhole-bacap-primer}

We summarize the three cryptographic building-blocks that \textit{funion} inherits from Echomix:  

\paragraph{Sphinx\,\&\,SURBs (routing layer).}
Every message that enters the mixnet is a constant-length \emph{Sphinx} packet~\cite{danezis2009sphinx}: an onion-encrypted header plus an encrypted payload.  
A sender chooses a $k$-hop route, wraps the packet in $k$ layers of public-key encryption, and attaches a \emph{Single-Use Reply Block (SURB)}-another onion header that encodes a \emph{return path}, encrypted such that neither the recipient nor any intermediate hop can link it back to the sender.  
When the recipient later replies, the SURB is consumed, guaranteeing unlinkability and non-replay.  
Packets that carry data and a SURB are colloquially called \textit{echoes}; packets without a SURB form the loop-cover traffic that each client emits continuously, so that an external observer cannot tell the two apart.

\paragraph{Pigeonhole storage (service layer).}
On top of Sphinx, Echomix offers a stateless \emph{courier} API: clients deposit or fetch opaque blobs at replica servers through echoes.  
Each blob sits in a "Pigeonhole" identified by a pseudorandom 32-byte string, its \textbf{Box-ID}.  
Uploads, downloads, and replica-to-replica gossip are all ordinary Sphinx packets, so a global passive adversary sees nothing but cover traffic.

\paragraph{BACAP vanilla (cryptographic core).}
BACAP (Blinding-And-Capability) deterministically turns one 256-bit seed into an infinite one-way chain of storage locations and keys:

\[
\begin{aligned}
H_i,E_i,K_i &= \mathrm{KDF}(H_{i-1}, i)\\
K^{\mathrm{ctx}}_i &= \mathrm{KDF}(K_i,\mathrm{ctx})\\
M^{\mathrm{ctx}}_i &= P_R \!\cdot\! K^{\mathrm{ctx}}_i
  \quad(\text{Box-ID})\\
S^{\mathrm{ctx}}_i &= S_R\,K^{\mathrm{ctx}}_i \bmod \ell
\end{aligned}
\]

\noindent
where $B$ is the Ed25519 base point and $\ell$ its prime order.  
Each piece of data that lands in a Pigeonhole is the \textbf{signed triple}

\begin{center}
\(\displaystyle(M,\,c,\,s)\)
\end{center}

\noindent
produced as shown in Table~\ref{tab:bacap-triple}.  Because a Box-ID is the
product of a \emph{public} point and a pseudorandom scalar, it is
computationally indistinguishable from a fresh Ed25519 public key; different boxes appear unlinkable.

\begin{table}[h]
\centering
\caption{The BACAP record format}
\label{tab:bacap-triple}
\begin{tabular}{@{}p{0.14\linewidth}p{0.12\linewidth}p{0.63\linewidth}@{}}
\toprule
\textbf{Field} & \textbf{Symbol} & \textbf{How it is derived} \\ \midrule
Box-ID & \(M\) &
\(M_i^{\text{ctx}}\!=\!P_R\!\cdot\!K_i^{\text{ctx}}\) \\[0.2em]
Ciphertext & \(c\) &
\(c_i^{\text{ctx}}\!=\!\text{AES-256-GCM-SIV}(m_i,\,E_i^{\text{ctx}})\) \\[0.2em]
Signature & \(s\) &
\(s_i^{\text{ctx}}\!=\!\text{Ed25519-SIGN}(c_i^{\text{ctx}},\,S_i^{\text{ctx}})\) \\ \bottomrule
\end{tabular}
\end{table}

\vspace{-0.2em}

BACAP cleanly separates \emph{write} and \emph{read} authority:

\begin{description}[leftmargin=1.4em,itemsep=0.3em]
\item[Write capability] \(W=(S_R,H_0)\) 
lets the holder \emph{create} new boxes by signing fresh triples \((M,c,s)\).
\item[Read capability]  \(R=(P_R,H_0)\) 
lets the holder \emph{enumerate} the same Box-ID sequence, decrypt each \(c\)
with \(E_i^{\text{ctx}}\), and verify \(s\) - but \emph{cannot} forge writes.
\end{description}

\noindent
The Pigeonhole storage protocol features four operational roles-each equipped with only the
minimum capability it requires:

\begin{description}[leftmargin=1.6em,itemsep=0.4em]
  \item[Client.] Generates fresh \(W\) or \(R\), selects couriers and replicas, and wraps every request in a Sphinx echo. A client may delegate \(R\) to other principals but never discloses its private \(W\).
  \item[Courier.] A stateless relay that terminates client echoes and forwards opaque \emph{envelopes} to the chosen replicas over a fixed-rate side channel. The courier sees timing and size only-it never learns the Box-ID.
  \item[Replica (storage server).] A key–value store indexed by Box-ID. It accepts a write when the signature verifies under the public key implicit in \(M\) and gossips records to peer replicas so that any \(k\)-of-\(n\) subset can satisfy a read.
  \item[Mix node.] Part of the Echomix transport fabric; it handles only fixed-length, onion-encrypted packets and contributes the exponential per-hop delay that gives the network its memoryless-mixing property.
\end{description}

\vspace{-0.9em}

\begin{intuitionbox}[Pigeonhole message flow]
A client \textbf{writes} by sending \((M,c,s)\) under a fresh
\(W\) through a courier to its designated replicas.
Anyone holding the corresponding \(R\) can later \textbf{read} the
same chain of boxes but cannot alter them.
Because no external party sees both the user's network identity
\emph{and} the Box-ID, unlinkability is preserved.
\end{intuitionbox}

\section{System Definition}
\label{sec:system-definition}

\textit{funion} operationalizes the insight that we can treat a mix network as a secure pathway to stateful service nodes that both store input tensors and perform computations on them. Our design builds directly on Echomix's Pigeonhole storage and BACAP protocol, recasting neural inference as "Pigeonhole-style store $\rightarrow$ compute $\rightarrow$ store" to inherit security guarantees.

\subsection{System Architecture and Components}

\textit{funion} consists of three primary component types, each with specific roles and trust assumptions (Figure \ref{fig:funion-architecture}):

\begin{description}
  \item[User] ($u \in \mathcal{U}$):  
        Performs critical anonymous operations locally, including: (i) local tokenization of input $x \in \mathcal{X}$, (ii) wrapping requests into Sphinx packets, and (iii) specifying valid mix network routes. The user interacts with the mixnet through a gateway node.

  \item[Mixnet] ($N \in \mathcal{N}$):  
        An Echomix network consisting of gateway nodes, mix nodes arranged in three layers, and service nodes. The mixnet routes packets between users and service nodes while hiding metadata through memoryless mixing, cover traffic, and layered encryption.

  \item[Service Node] ($s \in \mathcal{S}$):  
       A node at the service layer of the mixnet that processes Sphinx echoes. \textit{funion} employs two specialized service types:
       \begin{description}
         \item[Storage Courier] (Bob, Ben): 
           A service node that forwards encrypted BACAP operations to a swarm of $k$-of-$n$ storage replicas chosen by consistent hashing for each Box-ID. Couriers never see Box-IDs themselves, as these are encrypted within envelopes addressed to the replicas.
         \item[Compute Courier] (Charlie):
           A specialized courier that additionally performs neural inference on tensors retrieved from replicas. It acts as a client with respect to the storage couriers, creating standard read and write envelopes.
       \end{description}
       
  \item[Replica] ($r \in \mathcal{R}$):  
       Storage nodes positioned outside the mixnet that implement the actual key-value store. Each Box-ID is deterministically mapped to $k$ replicas using Pigeonhole's consistent hashing scheme. Replicas verify BACAP operations and maintain the tensor data.
\end{description}

\subsection{BACAP for Neural Inference}

For neural network inference, we make the following extensions to BACAP:

\begin{itemize}[leftmargin=*]
  \item \textbf{Separate read/write capabilities}: Alice generates a write capability $W_{\text{in}}$ for herself but only shares a read capability $R_{\text{in}}$ with Charlie. This ensures Charlie can read but not forge inputs.
  
  \item \textbf{Fresh capabilities for outputs}: Alice generates a new write capability $W_{\text{out}}$ and corresponding read capability $R_{\text{out}}$. She keeps $R_{\text{out}}$ and gives $W_{\text{out}}$ to Charlie to use for storing computation results.
  
  \item \textbf{Different contexts for input/output}: By binding different context values (e.g., $\mathrm{ctx}_{\text{in}}$, $\mathrm{ctx}_{\text{out}}$) to input and output sequences, even a global adversary cannot match the two blinded ID streams.
\end{itemize}

\subsection{Pigeonhole Inference}

Under Echomix's Pigeonhole storage protocol and BACAP scheme, we can treat neural inference as a sequence of storage and computation operations, all mediated through the anonymizing mixnet. A client first writes tensors into "input Pigeonholes" through Bob, hands Charlie a read-only ticket for those boxes plus a write-ticket for a fresh "output Pigeonhole" set, and finally reads the results back from Ben. Every communication step is an ordinary Echomix echo, so a global passive adversary sees five fixed-size packets and nothing else. The protocol operates as follows:

\begin{enumerate}[leftmargin=*]
  \item \textbf{Upload}: Alice splits her input tensor into fixed-size chunks and uploads them to a randomly chosen storage service (Bob) using a BACAP write capability $W_{\text{in}}$ with context $\textrm{ctx}_\textrm{in}$. Each upload is an ordinary Echomix echo that is indistinguishable from cover traffic.
  
  \item \textbf{Dispatch}: Alice randomly selects a compute service (Charlie) and sends an envelope containing a read capability $R_{\text{in}}$ for the input tensors and a freshly generated write capability $W_{\text{out}}$ for the output with context $\textrm{ctx}_\textrm{out}$. This is delivered via a standard SURB echo.
  
  \item \textbf{Compute}\textsubscript{fetch}: Charlie uses the read capability to fetch the input tensors from Bob through the mixnet, introducing another round-trip echo that preserves unlinkability.

  \item \textbf{Compute}\textsubscript{store}: After performing the neural network inference locally, Charlie stores the results on another randomly chosen service (Ben) via a full mixnet echo using the write capability $W_{\text{out}}$.

  \item \textbf{Fetch}: Alice retrieves the results from Ben using the read capability $R_{\text{out}}$ that corresponds to the write capability she gave to Charlie.
\end{enumerate}

\begin{table}[ht]
\caption{Pigeonhole Inference}
\label{tab:pipeline}
\resizebox{0.49\textwidth}{!}{
\begin{tabular}{@{}p{0.18\columnwidth}p{0.4\columnwidth}p{0.35\columnwidth}@{}}
\toprule
\centering
\textbf{Echo} & \textbf{Operation} & \textbf{Security Properties} \\
\midrule
Upload &
Alice $\rightarrow$ \textbf{Bob}: \hfill
encrypted input chunks written with $W_{\text{in}}$ &
Indistinguishable from decoy traffic \\

\midrule
Dispatch &
Alice $\rightarrow$ \textbf{Charlie}: \hfill
envelope with $R_{\text{in}}$ and $W_{\text{out}}$ &
Indistinguishable from decoy traffic \\

\midrule
{Compute}\textsubscript{fetch} &
Charlie $\leftrightarrow$ \textbf{Bob}: \hfill
pull input tensors with $R_{\text{in}}$ from replicas with $\textrm{ctx}_\textrm{in}$ &
Charlie learns no client ID; replicas see no Box-ID link \\

\midrule
{Compute}\textsubscript{store} &
Charlie $\leftrightarrow$ \textbf{Ben}: \hfill
store results with $W_{\text{out}}$ to replicas with $\textrm{ctx}_\textrm{out}$ &
Outputs unlinkable from inputs \\

\midrule
Fetch &
Alice $\rightarrow$ \textbf{Ben}: 
retrieve results from replicas with $\textrm{ctx}_\textrm{out}$ using $R_{\text{out}}$ &
Readers unlinkable from writers \\
\bottomrule
\end{tabular}}
\end{table}

\paragraph{Charlie's Store$\rightarrow$Compute$\rightarrow$Store Loop.}
At the core of the protocol is Charlie's computation process:

\begin{algorithm}[H]
\small
\caption{Charlie's inner loop for one inference job}
\begin{algorithmic}[1]
\State $\mathbf{x}, t_j \gets\text{Fetch}(R_{\text{in}})$  \Comment{via Bob through mixnet}
\State $\mathbf{y}\gets \mathcal{F}_\theta(\mathbf{x})$          \Comment{local forward pass}
\State $\text{WaitForBucketEdge}({t_j, \mathbf{y}})$ \Comment{Algorithm \ref{alg:bucket}}
\State \text{Store}$(W_{\text{out}},\mathbf{y})$     \Comment{to Ben through mixnet}
\end{algorithmic}
\end{algorithm}

\paragraph{Why Charlie tunnels via the mixnet.} Crucially, Charlie's communication with both Bob and Ben occurs through the mixnet rather than directly. Charlie never learns replica IDs because envelopes are end-to-end encrypted to the replicas. This preserves the uniform wire-image of all traffic and maintains the three-party trust split (Charlie, courier, replica) that our security reduction relies upon.

Every step in this pipeline is implemented as a standard Echomix echo and every piece of stored data is a BACAP box. The split-capability pattern established here forms the foundation for our IO-U $\Rightarrow$ SRTU reduction in \S\ref{sec:security}.

\subsection{Latency-Bucket Release Policy}

The latency-bucket release policy is an additional safeguard that does not contribute to Sender–Receiver Third-Party Unlinkability (SRTU). Its sole purpose is to blunt a separate side channel: an on-path observer who watches Charlie's packet timings (but cannot see payloads, identify endpoints, or compromise Charlie) might otherwise infer query characteristics from raw execution time. By forcing all visible wall-clock delays into coarse, publicly advertised buckets, we tackle the computational-time unlinkability threat highlighted by \cite{carlini2024remote}. This countermeasure acts at the granularity of network-level message release; micro-architectural leaks such as L1/L2 cache or CPU-cycle timing channels are outside its concern and, in practice, are already drowned out by the high variability of modern neural-inference workloads (kernel-launch jitter, payload-specific optimizations, tensor-shape shifts, GPU scheduling noise, etc.).

\begin{enumerate}[leftmargin=*]
  \item \textbf{Client-chosen bucket}: Upon creating a compute request, the client Alice chooses a latency bucket $t_j$ from a public grid of times $0=t_0 < t_1 < \dots < t_n$, evenly spaced by $\Delta$. This represents her expectation of how long the computation should take.
  
  \item \textbf{Service commitment}: The service Charlie promises to complete the computation before time $t_j$. This commitment is based on the public parameters of the request.
  
  \item \textbf{Timing discipline}: When computation completes, Charlie waits until the grid edge $t_j$ before releasing the result. This ensures that the observable release time is quantized to the public grid.
  
  \item \textbf{Handling overruns}: If the computation doesn't complete by time $t_j$ (the client-chosen bucket), the service returns a failure/overflow message, instructing Alice to resubmit with a larger bucket selection.
\end{enumerate}

This graduated-bucket approach preserves timing guarantees: all timing information visible to adversaries remains a deterministic function of public metadata (the chosen bucket index $j$) with no dependence on private computation characteristics.

\section{Security Analysis}
\label{sec:security}

Our security analysis relies on a direct reduction to the established security properties of Echomix and BACAP.

\subsection{Modeling Assumptions}
\begin{description}
  \item[Memoryless mixing.] \label{asm:echomix}
    All gateways, mixes and services adhere to Echomix's timing discipline: messages leave each participant after an independent, exponentially distributed $\text{Exp}(\lambda)$ delay; clients emit application and loop packets according to $\text{Pois}(\lambda_s)$ processes.
  
 \item[Path independence.]  \label{asm:path}
    At every layer of the mixnet topology, the next hop is selected independently of all previous routing decisions, and the resulting choice is computationally indistinguishable from a uniform draw over the nodes in that layer.

  \item[Constant packet size.] \label{asm:packet}
    All observable packets have the \emph{exact} Sphinx echo size, making them indistinguishable to a passive observer.
  
  \item[Capability freshness.] \label{asm:capability}
    Each inference uses a freshly generated $\mathcal{W}_{\text{in}}$, $\mathcal{W}_{\text{out}}$. Re-using a write seed would leak the public root key $P_R$ and break BACAP unlinkability.
  
  \item[Padding.] \label{asm:bucket-padding}
    For public parameters $\Delta$ and latency grid $0 = t_{0} < t_{1} < \dots < t_{n}$, each inference advertises a bucket index $j$. The service releases the result at the first edge $t_{m} \ge t_{\mathrm{finish}}$. If $t_{\mathrm{finish}} \ge t_{j}$, it instead returns a special OVERFLOW marker (sent at $t_{j}$).

    Hence the adversary's observable state over $n$ slots becomes
    \[
        (n,\,
          \langle (b_{1},o_{1}),\dots,(b_{n},o_{n})\rangle),
    \]
    where each slot now contains the public bucket index $b_{s}$ \emph{and}
    a one-bit overflow flag $o_{s}\in\{0,1\}$.
  
  \item[No early release.] \label{asm:early}
    Honest services never release results before a bucket edge even if the computation is finished earlier.
  
   \item[Limited courier–replica collusion.] \label{asm:collusion}
        For every inference job either the compute courier Charlie is honest or none of the $k$ replicas that store that job under $\mathrm{ctx}_{\text{in}}$ or $\mathrm{ctx}_{\text{out}}$ collude with him.

    \item[Non-critical mix corruption.] \label{asm:critical}
        The adversary may compromise any subset of gateways, mix nodes, couriers, and replicas except those critical combinations that simultaneously observe (i) a client's ingress identity and (ii) the replica set that holds the same client's Box-ID sequence. The directory-authority quorum contains at least one honest key and provides a consistent directory view to all honest participants.
    
    \item[Self-receiver model.] \label{asm:self-receiver}
        Each logical sender is the sole authorized receiver of the corresponding output; delegation of fetch rights is outside the scope of the present proof.

    \item[Gateway cover-traffic floor.]  \label{asm:cover}
        To keep every gateway-mix link active with high probability during one mean-delay interval $\mu$, each gateway must emit $\Theta(n^{2}\!\log n / g)$ packets in that time, where $n$ is the layer size of the mixnet and $g$ the number of gateways (Echomix's Coupon-Collector bound).
\end{description}

\begin{lemma}[Bucket-edge timing leaks $\leq$ 1 extra bit] \label{lemma:bucket-leakage} Let $o \in \{0,1\}$ be the overflow flag returned by the service ("1" iff $t_{\mathrm{finish}} \geq t_j$). The adversary's observable tuple is $(o, t_{\mathrm{release}})$. Given the public bucket index $j$ and grid ${t_0,\ldots,t_n}$, $(o, t_{\mathrm{release}})$ is a deterministic function of $(j, o)$, so the view leaks at most one additional bit (the value of $o$) beyond $j$. \end{lemma}

\begin{proof} If $o = 0$ the service releases exactly at the grid edge $t_j \geq t_{\mathrm{finish}}$; if $o = 1$ the service releases OVERFLOW. Hence timing is fully determined by $(j, o)$ and cannot encode anything else about the private running-time distribution. \end{proof} 

\begin{lemma}[Buckets + overflow leak $\le \log_2(n+1)$ bits]
With $n+1$ public grid edges and one possible overflow marker, the attacker’s uncertainty set has size $n+1$, so the Shannon leakage is bounded by $\log_2(n+1)$ bits.  Because the bucket index $j$ is already public, this reduces to a single extra bit (the overflow flag).
\end{lemma}

\subsection{Security Definition}

\begin{lemma}[Self-receiver IO-U $\Rightarrow$ SRTU] \label{lemma:io-srtu} Consider a \textit{funion} deployment in which every message is eventually fetched \textbf{only by the sender that created it} ("self-receiver model"). Under this condition, any adversary that breaks SRTU can be turned into an adversary that breaks IO-U with the same advantage.

\begin{proof} Let $\mathcal{A}$ be an adversary that wins the SRTU game with advantage $\varepsilon$. Because each logical \emph{job} has exactly one sender and that same principal is the unique authorized fetcher, we can embed $\mathcal{A}$ as the environment inside the IO-U game: \begin{enumerate}[leftmargin=*]
\item The reduction simulates the entire network for $\mathcal{A}$, forwarding every sender's \textbf{two} challenge jobs $(\text{in}_0,\text{out}_0)$ and $(\text{in}_1,\text{out}_1)$ to its own IO-U challenger. 
\item Which input the challenger actually processes determines the concrete receiver-side traffic pattern (because the receiver equates sender). 
\item $\mathcal{A}$ outputs a bit $b'$ guessing which sender/receiver pair was used. The reduction outputs the same bit as its IO-U guess. 
\end{enumerate} Since the simulation is perfect in the self-receiver model, the reduction wins the IO-U game with the same $\varepsilon$. \end{proof} \end{lemma}

The implication holds for any adversary that observes at least one end-point event (upload or fetch); all GPA adversaries do so by definition.

\begin{definition}[Input-Output Unlinkability, restated]
Let two honest clients $S_0,S_1\in\mathcal{C}$ each submit one inference
job, resulting in output box sequences
$\{N^{(0)}_j\}_j$ and $\{N^{(1)}_j\}_j$ on Ben.
The challenger flips $b\!\in\!\{0,1\}$ and tells Ben to deliver
$\{N^{(b)}_j\}$ to $S_0$ and $\{N^{(1-b)}_j\}$ to $S_1$.
A GPA outputs a bit $b'$.  
The advantage is $|\Pr(b'=b)-\tfrac12|$.
\end{definition}

This is the usual Echomix unlinkability game applied to
\emph{pairs} of echoes (upload + fetch) rather than single messages,
and strictly implies the standard sender-receiver unlinkability property,
as stated in Lemma~\ref{lemma:io-srtu}.

\begin{intuitionbox}
\noindent Think of the two "messages" in the game as:  \\
1) the "upload" of an input tensor, and  \\
2) the "download" of its processed output. 
Breaking IO-U would allow an adversary to pair them together, revealing
which input produced which output.
\end{intuitionbox}

\begin{theorem}[\textit{funion} inherits Echomix\,+\,BACAP anonymity]
\label{thm:reduction}
Let $\varepsilon$ be the maximum advantage of any probabilistic
polynomial-time global passive adversary (GPA) in the
IO-U game defined above.
Assume the following parameters of the two building blocks:

\begin{equation*}
\begin{aligned}
\varepsilon_E &\;=\; \text{advantage against \emph{one} Echomix echo}, \\
\delta &\;=\; \text{advantage to link \emph{one} BACAP box pair}.
\end{aligned}
\end{equation*}

Then, under Assumptions~\ref{asm:echomix},  

\[
\boxed{\;
      \varepsilon \;\le\; 4\,\varepsilon_E + \delta
      \;}
\]

In other words, any GPA that distinguishes \textit{funion}'s real world from the
ideal world with advantage $\varepsilon$ can be transformed into  

\begin{enumerate}[leftmargin=*]
\item an adversary that breaks the anonymity of \emph{one} Echomix echo with
  advantage $\ge \varepsilon/4$, or
\item  an adversary that links two BACAP records with advantage
  $\ge \varepsilon$.
\end{enumerate}

Hence \textit{funion} leaks no more information than the sum of four independent
Echomix echoes plus one BACAP box.  When $\varepsilon_E$ and $\delta$
are negligible, so is $\varepsilon$.
\end{theorem}

\subsection{Proof Sketch}

We prove by a standard hybrid game
argument~\cite{hybrid, goldreich2001foundations}.  
Starting from the real execution, we scrub, step by step, every piece of
information that an adversary could use, until only uniformly random
Sphinx traffic remains.

\paragraph{Five Sphinx echoes per inference.}
For clarity we name the round trips that occur in one \textit{funion} job:

\begin{center}
\renewcommand{\arraystretch}{1.05}
\resizebox{0.45\textwidth}{!}{
\begin{tabular}{@{}cllll@{}}
\toprule
Echo & Initiator & Outbound & Inbound & Purpose \\ \midrule
E1 & Alice   & \(A\!\rightarrow\!B_{\text{store}}\) & \(B_{\text{store}}\!\rightarrow\!A\) & \textbf{Upload} input tensor \\
E2 & Alice   & \(A\!\rightarrow\!C\) & \(C\!\rightarrow\!A\) & \textbf{Dispatch} \((R_{\text{in}},W_{\text{out}})\) \\
E3 & Charlie & \(C\!\rightarrow\!B_{\text{store}}\) & \(B_{\text{store}}\!\rightarrow\!C\) & \textbf{Fetch} input with \(R_{\text{in}}\) \\
E4 & Charlie & \(C\!\rightarrow\!B_{\text{result}}\) & \(B_{\text{result}}\!\rightarrow\!C\) & \textbf{Store} result with \(W_{\text{out}}\) \\
E5 & Alice   & \(A\!\rightarrow\!B_{\text{result}}\) & \(B_{\text{result}}\!\rightarrow\!A\) & \textbf{Fetch} result with \(R_{\text{out}}\) \\ \bottomrule
\end{tabular}
}
\end{center}

Only E3 and E4 are \emph{started by Charlie}; each of them has a reply
that flows \emph{into} Charlie.  
Those four Charlie-visible packets are the ones that can reveal extra
routing information and thus need to be masked in the proof.

\medskip
Let \(\varepsilon_E\) be the anonymity advantage against \emph{one}
Echomix echo (from the Echomix security theorem) and let
\(\delta\) be the unlinkability bound for one BACAP box pair.

\paragraph{Hybrid \(\boldsymbol{H_0}\) (real world).}
The genuine \textit{funion} execution, with all BACAP triples and Charlie's
computation intact.

\paragraph{Hybrid \(\boldsymbol{H_1}\) (hide \emph{replies} into Charlie).}
We replace the \emph{two} inbound packets  
\(B_{\text{store}}\!\rightarrow\!C\) (echo E3) and  
\(B_{\text{result}}\!\rightarrow\!C\) (echo E4)  
with fresh, uniformly random Sphinx echoes of the same length.
By the Echomix anonymity theorem each substitution changes the
adversary's view by at most \(\varepsilon_E\); a union bound gives  
\(\lvert \Pr(\mathcal{A}\text{ distinguishes }H_0,H_1)\rvert \le 2\varepsilon_E\).

\paragraph{Hybrid \(\boldsymbol{H_2}\) (hide \emph{outbound} packets from Charlie).}
Next, we randomize the \emph{two} outbound packets
\(C\!\rightarrow\!B_{\text{store}}\) (E3) and
\(C\!\rightarrow\!B_{\text{result}}\) (E4) in exactly the same manner,
delivering them to the mix-net at the same bucket-edge time
(Assumption \ref{asm:bucket-padding}).  
Again, Echomix anonymity bounds the distinguishing advantage by another
\(2\varepsilon_E\).

\paragraph{Hybrid \(\boldsymbol{H_3}\) (random BACAP triples).}
Reveal all symmetric keys and replace every BACAP record
\((M,c,s)\) with random bits of equal length.
BACAP's unlinkability property guarantees that the adversary's advantage
drops by at most \(\delta\).

\paragraph{Hybrid \(\boldsymbol{H_4}\) (ideal world).}
After the previous steps every observable packet is now an
independent, uniformly random Sphinx echo.
Hence the adversary's view is identical whether the challenge bit
\(b=0\) or \(b=1\); she can do no better than guess, i.e. her advantage
is \(0\).

\medskip
\noindent
Collecting the losses across hybrids we obtain
\[
\varepsilon \;\le\;
\underbrace{2\varepsilon_E}_{H_0\rightarrow H_1}
\;+\;
\underbrace{2\varepsilon_E}_{H_1\rightarrow H_2}
\;+\;
\underbrace{\delta}_{H_2\rightarrow H_3}
\;=\;
4\,\varepsilon_E + \delta .\;
\square
\]

\subsection{Capabilities of Colluding Pairs of Compromised Network Elements}
\label{sec:colluding-pairs}

Mirroring the presentation style of Echomix\,\cite[\S6.2]{infeld2025echomix}, we enumerate
what pairs of \textit{funion} components could learn if both are \emph{actively}
malicious.  A global passive adversary (GPA) is always assumed.

\begin{itemize}[leftmargin=*,itemsep=4pt,parsep=0pt]
\item \textbf{Gateway + Mix layer(s)}: Both only see uniformly padded, cover-mixed Sphinx echoes.  Because service
nodes are chosen independently of traffic history and payloads remain
indistinguishable, they cannot separate \textbf{Upload},
\textbf{Dispatch}, \textbf{Compute}, or \textbf{Fetch} echoes from loop cover
traffic.  No practical linkability is obtained.

\item \textbf{Gateway + Storage Courier.} A courier can batch up the client's repeated SURB-bearing retries and release them in one burst; the gateway sees that burst and therefore learns when a given client initiated a copy/write request and roughly how large it was. Because the envelope that carries the request is still end-to-end encrypted to replicas, neither party learns which Box-IDs (or which replicas) are involved, so they get timing/volume information only.

\item \textbf{Gateway + Compute Courier}: Charlie observes a \textbf{Dispatch} envelope plus the ensuing compute workload
bucket; the gateway sees when that envelope left a client. Jointly they can say "client $A$ originated a job using latency-bucket $j$," thereby learning per-client workload volume. They \textbf{still cannot} link to the specific input or output Box-IDs (protected by envelope encryption to replicas), so SRTU and IO-U remain unbroken.

\item \textbf{Gateway + Replica}:  The replica can see every time a particular Box-ID is read,
and the gateway can see which clients are active at those moments.
By comparing these two timelines over many days, they may
eventually guess which client owns a given Box-ID.
This "intersection" attack is slow, because each read is
routed through random couriers and released only at
coarse, pre-set time steps, so the timing clues are fuzzy.

\item \textbf{Mix layer(s) + Courier/Replica}: \begin{itemize}[leftmargin=*]
\item Mix + Courier.
      A compromised mix can tag or delay packets that head to
      one specific courier and later watch them return,
      linking those packets to a particular gateway
      (but not to an individual user).  
      It still cannot see Box-IDs, so the courier's view remains
      "opaque envelopes plus timing".
\item Mix + Replica.  
      The mix learns one hop of the path; the replica sees
      Box-IDs.  Because at least one other mix hop is honest, the
      pair cannot connect ingress user traffic to any specific
      Box-ID or client—they only learn that "some user of this
      gateway wrote/read this Box-ID."
\end{itemize}

\item \textbf{Storage Courier + Compute Courier}:  These two services can line up when Charlie fetched an input (from the storage courier) with when he later stored
the result (via the compute courier's outbound write),
effectively pairing the input-fetch echo E3 with the
result-store echo E4 for a single job.  
That reveals an \emph{input-output link} for that job's data
flow, but because neither courier sees client IPs or Box-IDs,
they still lack both the user's identity and the blinded storage
locations.  Charlie rotation and per-job replica randomness keep
this leakage bounded to the job being processed.

\item \textbf{Compute Courier + Replicas}: This is the \emph{strongest} collusion. Charlie holds plaintext tensors after retrieval, and the replicas know the Box-IDs and can link different BACAP operations on the same sequence. Jointly they can observe which input produced which output, revealing the \emph{input-output link}. \textit{funion} therefore requires Assumption~\ref{asm:collusion} and rotates compute services per job to make repeated collusion statistically unlikely.

\begin{intuitionbox}
 Crucially, however, they still lack the user's ingress identity: any real read is still buried among thousands of messages via the rest of \textit{funion}'s mix mechanisms. Because each inference randomly selects a fresh courier and a fresh \(k\)-of-\(n\) replica subset, the adversary must repeatedly compromise the correct nodes and gather many aligned observations before intersection or volume analysis shrinks the anonymity set, rendering practical deanonymization statistically costly even though the formal IO-U \(\Rightarrow\) SRTU guarantee no longer holds.
\end{intuitionbox}

\item \textbf{Replicas for both contexts}: Even if one physical operator stores every input $\textrm{ctx}_\textrm{in}$ record and every output $\textrm{ctx}_\textrm{out}$ record, the two
Box-ID sequences are generated under different BACAP contexts and
therefore look like unrelated random Ed25519 keys.  Linking them
requires breaking BACAP's one-way chain, so the attacker's advantage
is at most~$\delta$ (negligible).
\end{itemize}

Key takeaways:
\begin{itemize}[leftmargin=*,itemsep=2pt]
  \item No \emph{single} service role plus network-level vantage leaks both ends of an inference request
  \item The strongest risk (Compute Courier + Replicas collusion) is mitigated by Assumption~\ref{asm:collusion}
  \item \textit{funion} does not widen the adversarial surface beyond Echomix; therefore its resistance to the standard catalog of mixnet active attacks—including the {$n\!-\!1$ attack}, {Sybil attack}, and {intersection / statistical-disclosure attacks}—is exactly that of Echomix.  The security bounds and mitigation surveyed in \cite{infeld2024mixnet} apply mostly unchanged, and we shall observe it empirically in future works targeting real-world deployment.
\end{itemize}

\section{Performance Estimation}
\label{sec:performance}

We attempt to quantify \textit{funion}'s estimated overhead in latency and packet size.  All timing
figures combine (i) published Echomix parameters and (ii) vendor-supplied
benchmarks for Llama-3-70B inference.

\subsection{Packet Capacity and Efficiency}
\label{subsec:packet-capacity}

We consider the scenario where a single language model is served by all instances of Charlie in the mixnet. In this case, a full neural inference request fits into a \emph{single} Sphinx packet, so no
fragmentation is needed.  Katzenpost's reference implementation fixes

\begin{itemize}[leftmargin=*]
  \item \textbf{Sphinx payload capacity:} 30\,000 bytes
  \item \textbf{Token representation:} 4 bytes/token (32-bit integer)
  \item \textbf{Maximum prompt:} $\lfloor30\,000/4\rfloor = 7\,500$ tokens
\end{itemize}

Hence even a 5 000-token prompt (\(\approx\) 20 kB) leaves ample room for the
BACAP envelope and padding.

\subsection{Inherited Mixnet Parameters}
\label{subsec:mix-params}

We use values in Echomix's deployment by Zero Knowledge Network unchanged:

\begin{table}[h]
\centering
\begin{tabular}{lll}
\toprule
\textbf{Parameter} & \textbf{Value} & \textbf{Meaning} \\
\midrule
$\mu$              & 0.20 s & Mean delay per hop \\
$k$                & 3      & Mix depth (one way) \\
$\lambda_s$        & 2.5 pkt/s & Client loop-cover rate \\
\bottomrule
\end{tabular}
\end{table}

\paragraph{One echo.}  The full round-trip timing follows an 
Erlang\((k{=}9,\lambda{=}5\,\mathrm{s}^{-1})\) distribution with:  
\(\mathrm{E}[X]=9/5=1.8\) s,\;
\(\mathrm{Var}[X]=9/25=0.36\;\mathrm{s}^2\).

\paragraph{Five echoes per inference.} 

A \emph{single} forward pass touches the mixnet five times:

\begin{enumerate}[leftmargin=*]
    \item \textbf{Upload} - Alice \(\rightarrow\) Bob: 1 Sphinx packet (\(\approx\)31 kB) carrying the input tensor encrypted under \(W_{\text{in}}\).
    \item \textbf{Dispatch} - Alice \(\rightarrow\) Charlie: 1 packet with the job ticket \((R_{\text{in}},\,W_{\text{out}})\).
    \item \textbf{Compute}\textsubscript{fetch} - Charlie \(\rightarrow\) Bob: 1 packet requesting the stored tensor via \(R_{\text{in}}\).
    \item \textbf{Compute}\textsubscript{store} - Charlie \(\rightarrow\) Ben: 1 packet that writes the result under \(W_{\text{out}}\).
    \item \textbf{Fetch} - Alice \(\rightarrow\) Ben: 1 packet that retrieves the output tensor with \(R_{\text{out}}\).
\end{enumerate}

\vspace{-1em}

\[
\mathrm{E}[X_{\text{mix}}]=5\times1.8=9.0\text{ s},\qquad
\mathrm{Var}[X_{\text{mix}}]=5\times0.36=1.8\text{ s}^2 .
\]

\subsection{LLM Inference Latency}
\label{subsec:llm-benchmarks}
Table~\ref{tab:nim} presents canonical inference latencies for Llama-3.3-70B-Instruct (fp16, tensor-parallel=4) on 4$\times$H100 GPUs using NVIDIA's NIM platform \cite{nvidia2025nim70b}, where: TTFT = Time To First Token; ITL = Inter-Token Latency; $t_{\text{LLM}} = \text{TTFT} + n_{\text{out}} \times \text{ITL}$.

\begin{table}[h]
\centering
\resizebox{0.45\textwidth}{!}{
\begin{tabular}{@{}lccccc@{}}
\toprule
\textbf{Scenario} & $n_{\text{in}}$ & $n_{\text{out}}$ & TTFT (ms) & ITL (ms) & $t_{\text{LLM}}$ (s) \\
\midrule
Balanced & 200 & 200 & 32.78 & 19.11 & 3.85 \\
Medium & 1000 & 1000 & 103.20 & 19.31 & 19.41 \\
Output-heavy & 500 & 2000 & 71.82 & 19.26 & 38.59 \\
Input-heavy & 5000 & 500 & 368.11 & 19.94 & 10.34 \\
\bottomrule
\end{tabular}}
\caption{Latency on 4$\times$H100 80GB (NIM, fp16, TP=4)}
\label{tab:nim}
\end{table}

\subsection{Latency Overhead}
\label{subsec:overhead-analysis}

Let \(t_{\text{LLM}}\) be the pure compute time and
\(t_{\text{mix}}=9.0\) s the expected networking delay. With bucket quantization using \(\Delta = 0.2\) s, the computation time is rounded up to the next bucket edge:
\[
t_{\text{LLM}}^{\text{rounded}} = \left\lceil \frac{t_{\text{LLM}}}{\Delta} \right\rceil \cdot \Delta
\]

the mix percentage is then:

\[
\rho=\frac{t_{\text{mix}}^{\text{rounded}}}{t_{\text{mix}}^{\text{rounded}}+t_{\text{LLM}}}.
\]

\begin{table}[h]
\centering
\resizebox{0.45\textwidth}{!}{
\begin{tabular}{@{}lcccc@{}}
\toprule
\textbf{Scenario} & \(t_{\text{LLM}}^{\text{rounded}}\) (s) & \(t_{\text{mix}}\) (s) & Total (s) & Mix \% \\
\midrule
Balanced (200/200) & 4.00  & 9.0 & 13.00 & 69 \% \\
Medium (1k/1k)     & 19.60 & 9.0 & 28.60 & 31 \% \\
Output-heavy       & 38.60 & 9.0 & 47.60 & 19 \% \\
Input-heavy        & 10.40 & 9.0 & 19.40 & 46 \% \\
\bottomrule
\end{tabular}}
\caption{End-to-end latency with 5 Sphinx echoes ($\mu = 0.20 s$)}
\label{tab:overhead}
\end{table}

Even for the 38 s output-heavy case, networking contributes less than one
fifth of total latency.  For shorter prompts, the privacy budget is paid
largely in delay; deployments that co-locate couriers and replicas could drop
the two Charlie-internal echoes and shrink \(t_{\text{mix}}\) to 5.4 s
(three echoes) at the cost of a weaker trust split.

The values in Table \ref{tab:overhead} represent a \emph{best-case} scenario because they are based on Llama-3.3-70B—one of the larger publicly benchmarked models. For smaller LLMs (e.g., 7 B–13 B parameters) whose forward pass often finishes in $\le\!1$ s on a single GPU, the fixed $t_{\text{mix}}\!\approx\!9$ s term would dominate the end-to-end latency, pushing the Mix \% well above 90 \% and, in the extreme, making network delay the primary cost of anonymity. The anonymity budget incurred by \textit{funion} grows with the volumes of the requests, but not the computational intensity of requests themselves. Therefore, it is more suited for large models. 

\subsection{Bandwidth Constraints}

\textit{funion} requires continuous cover traffic to maintain anonymity guarantees. Using Echomix's established parameters ($\lambda_s=2.5$ packets/second with 31 kB per packet), we calculate the baseline bandwidth requirement per client:

\begin{align*}
\text{Baseline bandwidth} &= 2.5 \text{ pkt/s} \times 31 \text{ kB/pkt} \times 86400 \text{ s/day} \\
&\approx 6.7 \text{ GB/day}
\end{align*}

This 6.7 GB/day represents the minimum bandwidth footprint for maintaining anonymity through loop-cover traffic for a single client, irrespective of actual inference usage. The X25519 NIKE Sphinx configuration in Echomix introduces 1 kB of overhead that we have already factored in for each inference query.

\paragraph{Per-Inference Costs.} From client's perspective (\S\ref{subsec:mix-params}), each complete inference requires 3 Sphinx packets, and consumes approximately 93 kB of bandwidth. Under a substitution model where application packets replace loop packets, a client with the minimum 6.7 GB/day bandwidth allocation could perform:

\begin{align*}
\text{Max inferences/day} &= \frac{\lambda_s \times \text{seconds/day}}{\text{packets/inference}} \\
&= \frac{2.5 \text{ pkt/s} \times 86400 \text{ s}}{3 \text{ pkt/inference}} \\
&\approx 72000 \text{ inferences/day} \\
&\approx 50 \text{ inferences/minute}
\end{align*}

\paragraph{Contextualization with ServeGen.} To understand whether this bandwidth constraint is practical, we analyze it in the context of ServeGen workload characteristics from \cite{xiang2025servegen}. ServeGen shows that request rates are highly skewed: 29 clients (1.2\%) are responsible for 90\% of traffic. The quiet majority send requests far less frequently than our 50-inference/minute cover-traffic budget, whereas the noisiest client (Client A in their analysis) bursts to $\sim$150 req/s (9,000 req/minute) - about 180$\times$ more than our budget.

This comparison reveals that \textit{funion}'s anonymity overhead is negligible for typical human users or chatbot interfaces that generate sporadic traffic. However, power users or API-driven applications that represent the top 1-2\% of clients in production environments would require either:

\begin{enumerate}[leftmargin=*]
    \item Higher bandwidth allocation during peak periods
    \item Request buffering/throttling mechanisms to smooth bursts
    \item Adjustments to the anonymity parameters ($\lambda_s$ or $\mu$) to trade off anonymity strength against throughput
\end{enumerate}

\section{Limitations and Future Work}

While \textit{funion} demonstrates that mix networks can effectively enable anonymous neural inference, several important limitations and directions for future research remain:

\subsection{Hidden-state Privacy}

Our design intentionally addresses network-level anonymity rather than representation-level privacy. The plaintext hidden activations processed by services represent a separate privacy domain that is orthogonal to the sender-receiver unlinkability that \textit{funion} successfully establishes.

A compromised service processing a model could extract the hidden state and run subsequent layers locally to generate outputs or apply techniques described by \cite{pal2025hidden} to attempt reconstruction of inputs. In our case analyzed above with no split of model inference across multiple providers, the input is directly available. This is an inherent limitation of our current approach.

Importantly, \textit{funion}'s store-compute-store architecture provides a foundational defense by allowing model sharding across independent service nodes-forcing adversaries to compromise multiple specific nodes to observe a complete model. To fully protect hidden states, future work could explore:

\paragraph{Private models.} Store partitions of models \cite{huang2019gpipe} on disjoint services and issue sequential \textbf{Compute} requests, allowing each service only observing part of the network.
    
\paragraph{Representation obfuscation via LoRA.} Even with a public model, low-rank adaptations \cite{hu2022lora} could create private computational pathways where hidden states remain meaningfully obfuscated \cite{liao2021information} such that it is hard, if not impossible, to reconstruct either the output or the original input.

\subsection{Verifiable Computation}
\label{sec:vclimit}

Recent work on verifiable neural network inference offers complementary defenses
that ensures a compute provider's honest computation in \textit{funion}.

\paragraph{Hash–based commitments.}
TOPLOC shows that locality-sensitive hashes over carefully-chosen intermediate
activations can catch model, prompt or precision tampering at \(\approx\!100\times\)
\emph{faster} than re-running the full inference while adding only
\(258\;\mathrm{B}/32\) tokens of storage \cite{ong2025toploc}.
Because verification is just a hash comparison, a client (or even Charlie) could
embed TOPLOC checks inside a \textit{funion} envelope with negligible overhead.

\paragraph{Succinct proof systems.}
Fully cryptographic approaches generate a zero-knowledge,
succinct proof that the claimed neural-network execution is correct:

\begin{itemize}[leftmargin=*]
  \item \textbf{Kaizen} \cite{kaizen} verifies training of deep neural networks.
  \item \textbf{zkLLM} \cite{sun2024zkllm} produces $<\!200$ kB inference proofs for 13B-parameter LLMs in under 15 min.
\end{itemize}

These schemes reduce the need to trust either Charlie or the replicas, but
today they impose prover overheads \(10^{2}\!-\!10^{4}\times\) the original
computation, circuit-translation constraints, and sometimes model-specific
optimizations.  Until prover cost reaches near-real-time parity, we view
zkSNARK-style verification as adjacent future work rather than a core
assumption for \textit{funion}.

\paragraph{Trusted-Execution Environments (TEE).}
Hardware-based security mechanisms like Intel SGX provide isolated execution environments with remote attestation capabilities, allowing clients to verify that computations ran on genuine, unmodified code \cite{intelsgx}. Several prototypes have demonstrated that neural network inference inside these protected enclaves is viable \cite{mo2024machine}.

\textit{funion} could leverage TEEs to run complete inference queries or just lightweight verification checks, providing hardware-backed assurance as an alternative to cryptographic proofs. However, TEEs introduce a hardware trust dependency and remain vulnerable to sophisticated side-channel attacks like SIGY \cite{sridhara2024sigy}. Additionally, current enclave memory constraints (e.g., SGX's 128/256 MB limit) are ill-suited for billion-parameter models, resulting in substantial performance penalties. 

TEEs and hash-based commitments currently represent a pragmatic compromise that trades formal security guarantees for practical deployability, a reasonable interim solution until fully succinct proofs become computationally feasible.

\subsection{Evaluation In The Wild}

While we have established the security properties of \textit{funion} through formal analysis and estimated its theoretical performance characteristics in \S\ref{sec:performance}, an actual implementation and empirical evaluation of the system deployed on real devices remains an important direction for future work. 

\noindent Key aspects that require empirical validation include:

\begin{enumerate}[leftmargin=*]

    \item \textbf{Implementation overhead}: Assessing the engineering complexity and deployment considerations that arise when implementing the complete protocol stack in production environments, as the underlying mixnet, Katzenpost/Echomix, is still under active development. 

    \item \textbf{Practical end-to-end latency}: Validating our theoretical latency estimations with measurements from actual system deployments, including the effects of real network conditions, computation variability, and compute heterogeneity not captured in our model.
    
    \item \textbf{Achievable throughput}: Determining the maximum number of inference operations that can be processed per unit time in practice.
    
    \item \textbf{System scalability}: Measuring how performance and anonymity metrics vary with model size, network complexity, and concurrent user load in deployed environments, particularly focusing on how active disruptions affect practical mixing quality.

\end{enumerate}

\subsection{Deployment Considerations}

Deploying \textit{funion} presents several practical challenges:

\begin{enumerate}[leftmargin=*]
    \item \textbf{Infrastructure costs}: Maintaining a distributed network of mix nodes and service nodes requires significant resources from multiple stakeholders, which must be balanced against the privacy benefits.
    
    \item \textbf{Network dynamics}: As nodes join or leave the network, ensuring consistent anonymity properties and performance requires careful management.
    
    \item \textbf{Incentive alignment}: Creating sustainable economic incentives for operating mix nodes and service nodes while preserving anonymity remains an open challenge.
\end{enumerate}

Future work could explore decentralized governance models and incentive mechanisms to address these deployment challenges.

\section{Conclusion}

We present \textit{funion}, a system for anonymous neural network inference that leverages the Echomix mix network to provide strong theoretical guarantees. By introducing a Pigeonhole-based store-compute-store approach with BACAP capabilities, we achieve sender-receiver unlinkability even against sophisticated adversaries.

As neural network inference increasingly moves to cloud and distributed settings, the need for privacy will only grow. \textit{funion} represents a first step toward fully anonymous neural inference, where neither the content of the query, the response, nor the fact that such a user made the query are exposed. Our future work will focus on quantifying the performance characteristics of this approach and extending it to address hidden state privacy.

\section*{Acknowledgment}

The Ritual team has been immensely helpful in reviewing early drafts of the work, especially Akilesh Potti, Arka Pal, and Praveen Palanisamy. Noah Levenson contributed to various iterations of the idea.

\noindent To Iris, my cat, for the company.

\bibliographystyle{plain}
\bibliography{references}

\appendix

\section{Position Statement}

Privacy is a fundamental human right. \textit{funion} embodies this principle while intentionally sidestepping the complex policy debates that properly belong in political discourse. We respect the unique positions of sovereign states and offer a technical solution that exists within-not above-these diverse regulatory frameworks.

\noindent Frontier models are scaling at an appalling rate. A growing consensus suggests we are approaching natural limits of available human-generated data for training \cite{villalobos2022will}. This scarcity has driven large corporations to deploy consumer applications with the plausible goal of mass-level surveillance for corporate gain. The integration of these models into critical services creates powerful incentives for data collection far beyond what is necessary for service delivery.

\noindent \textit{funion} is an experimental prototype specifically designed to disincentivize corporate malfeasance in language model inference. As history has shown, general-purpose anonymous communication networks carry risks of misuse. \textit{funion}'s narrow focus on neural network inference represents a targeted approach to a specific privacy concern.

\noindent The current lack of hidden state privacy may serve as an opportunity for responsible content monitoring without revealing user identities. Privacy and safety need not be mutually exclusive goals.

\noindent Mixnet deployment does not necessitate cross-border data flows. \textit{funion} can be instantiated within a single geographical region while maintaining similar anonymity guarantees. The architecture adapts to deployment models that respect geopolitical boundaries while providing meaningful privacy protections. Our technical contribution should not be interpreted as advocating for any particular deployment model that might conflict with local regulations.

\section{Notation}
\label{sec:notation}

\paragraph{Entity sets}
\begin{description}[leftmargin=2.3em,itemsep=0.2em]
  \item[$\mathcal{U}$] Set of users / clients.
  \item[$\mathcal{C}\subseteq\mathcal{U}$] Honest (non-malicious) clients.
  \item[$\mathcal{A}$] Adversarial parties (may include users or nodes).
  \item[$\mathcal{N}$] All mix nodes in the network.
  \item[$\mathcal{G}\subseteq\mathcal{N}$] Honest mix nodes ($g=|\mathcal{G}|$).
  \item[$\mathcal{S}$] Service nodes (couriers) in the mixnet.
  \item[$\mathcal{R}$] Replica servers that store data outside the mixnet.
\end{description}

\paragraph{Traffic parameters}
\begin{description}[leftmargin=*,itemsep=0.2em]
  \item[$\lambda$] Parameter of the exponential delay distribution at a mix.
  \item[$\lambda_s$] Client Poisson send rate.
\end{description}

\paragraph{BACAP notation (inherited from Echomix)}
\begin{description}[leftmargin=*,itemsep=0.2em]
  \item[$B$] Ed25519 base point, $\ell$ its prime order, and $\mathrm{ctx}$ the network-wide context.
  \item[$S_R\in\mathbb{Z}_\ell$, $P_R=B\!\cdot\!S_R$] Root private / public key for a message.
  \item[$H_i,K_i,E_i$] key derivation function (KDF) state, location-blinding factor, and symmetric payload key, respectively.
  \item[$M^{\mathrm{ctx}}_i=P_R\!\cdot\!K_i$] Box ID; $S^{\mathrm{ctx}}_i=S_R\,K_i\!\pmod{\ell}$ per-box signing secret key.
  \item[$R_{\mathrm{in / out}}$, $W_{\mathrm{in / out}}$] Public read / private write capabilities.
\end{description}

\paragraph{Latency-bucket timing}
\begin{description}[leftmargin=*,itemsep=0.2em]
  \item[$\Delta$] Public spacing between bucket edges.
  \item[$t_0 < t_1 < \dots < t_n$] Global grid of bucket edges advertised by the client.
  \item[$j$] Bucket index selected by the client when dispatching a job.
  \item[$t_{\mathrm{finish}}$] Wall-clock time at which the service finishes computing.
  \item[$t_j$] Bucket edge $\ge t_{\mathrm{finish}}$; the service releases results at $t_j$.
\end{description}

\paragraph{Negligible functions \cite{hybrid}}
\begin{description}[leftmargin=*,itemsep=0.4em]
  \item[$\negl(\lambda)$]
        A function $\negl : \mathbb{N} \!\to\! \mathbb{R}^{+}$ is \emph{negligible} if  
        for every polynomial $p : \mathbb{N} \!\to\! \mathbb{R}^{+}$ there exists an integer
        $\Lambda$ such that for all $\lambda \ge \Lambda$
        \[
          \negl(\lambda) \;\le\; \frac{1}{p(\lambda)}.
        \]
        Equivalently, $\negl(\lambda) = o(\lambda^{-c})$ for every constant $c>0$.
  \item[$\delta$/$\epsilon$]  
        A concrete placeholder for a negligible value; we often write $\delta = \negl(\lambda)$.
        For example, in the Ed25519 setting one obtains the concrete bound
        $\delta \le 2^{-128}$.
\end{description}

\paragraph{Probability}
\begin{description}[leftmargin=2.3em,itemsep=0.2em]
  \item[$\mathrm{Exp}(\lambda)$] Exponential distribution.
  \item[$\mathrm{Pois}(\lambda)$] Poisson process.
  \item[\mbox{$\mathrm{E}(\cdot)$}, \mbox{$\Pr(\cdot)$}] Expectation and probability operators.
\end{description}

\section{Latency Bucket Wait Algorithm}
\label{alg:bucket}

The pseudocode below illustrates the wall-clock wait mechanism for the latency-bucket release policy:

\begin{algorithm}[ht]
\caption{Latency Bucket Wait}
\begin{algorithmic}[1]
\Function{WaitForBucketEdge}{$t_j$, $result$}
    \State $t_{\text{finish}} \gets \text{CurrentTime()}$
    \If{$t_{\text{finish}} \geq t_j$}
        \State \Return 0 ("OVERFLOW")
    \EndIf
    \State \textbf{wait until} $\text{CurrentTime()} = t_j$
    \State SendResult($result$) \Comment{Send result via mixnet}
    \State \Return 1 ("OK")
\EndFunction
\end{algorithmic}
\end{algorithm}

\section{Neural Network Preliminaries}
\label{sec:prelim}

This section provides a formal treatment of feedforward neural networks and autoregressive models, providing context for readers who may not be familiar with the literature.

\subsection{Feedforward Neural Networks}

A feedforward neural network $\mathcal{F}$ is a parameterized function that maps inputs to outputs through a series of sequential transformations \cite{goodfellow2016deep}. Formally, we define $\mathcal{F}$ as $\mathcal{F}_{\theta}: \mathcal{X} \rightarrow \mathcal{Y}$, where $\mathcal{X}$ is the input space, $\mathcal{Y}$ is the output space and $\theta$ represents the network parameters, or weights. The defining characteristic of a feedforward network is its directional, acyclic flow of information.

The computation, sometimes referred to as inference, in $\mathcal{F}_{\theta}$ proceeds through $L$ sequential layers, where each layer $j \in \{1, 2, \ldots, L\}$ applies a transformation $f_{j}$ parameterized by $\theta_{j}$. If we denote the input as $x \in \mathcal{X}$, then the computation proceeds as follows:

\begin{align*}
h_1 &= f_1(x; \theta_1) \\
&\cdots \\
h_L &= f_L(h_{L-1}; \theta_L) \\
y &= h_L \in \mathcal{Y}
\end{align*}

Here, $h_{j} \in \mathcal{H}_{j}$ represents the hidden state or activation at the layer $j$, and $\mathcal{H}_{j}$ is the corresponding hidden space. The complete network can be expressed as the composition of these layer-wise functions:

\begin{equation*}
\mathcal{F}_{\theta}(x) = (f_L \circ f_{L-1} \circ \cdots \circ f_1)(x)
\end{equation*}

In modern neural networks, particularly those used for natural language processing, each layer $f_{j}$ is typically composed of multiple sublayers and operations such as attention \cite{vaswani2017attention}:

\begin{equation*}
f_{j}(h_{j-1}; \theta_{j}) = \text{Norm}(\text{SubLayer}(h_{j-1}; \theta_{j}^{\text{sub}}) + h_{j-1})
\end{equation*}

where $\text{Norm}$ is a normalization function and $\text{SubLayer}$ represent operations like self-attention or feedforward neural networks with non-linear activation functions. The parameters are usually multi-dimensional arrays, also called tensors.

\subsection{Self-Attention}

The self-attention sublayer, introduced in the Transformer architecture \cite{vaswani2017attention}, is a critical component of many language models that enables their emergent capabilities to focus on different parts of the input sequence. At its core, self-attention relies on matrix multiplications. The scaled dot-product attention, a key operation in self-attention, is defined as:

\begin{equation*}
\text{Attention}(Q, K, V) = \text{softmax}\left(\frac{QK^T}{\sqrt{d_k}}\right)V
\end{equation*}

where $Q \in \mathbb{R}^{n \times d_k}$, $K \in \mathbb{R}^{m \times d_k}$, and $V \in \mathbb{R}^{m \times d_v}$ represent the query, key, and value matrices respectively, and $d_k$ is the dimension of the keys. The scaling factor $\frac{1}{\sqrt{d_k}}$ prevents the dot products from growing too large in magnitude, which could push the softmax function into regions with extremely small gradients \cite{vaswani2017attention}.

In self-attention, the $Q$, $K$, and $V$ matrices are obtained by linear projections of the same input:

\begin{align*}
Q = XW^Q, 
K = XW^K,
V = XW^V
\end{align*}

where $X \in \mathbb{R}^{n \times d_{\text{model}}}$ is the input matrix and $W^Q \in \mathbb{R}^{d_{\text{model}} \times d_k}$, $W^K \in \mathbb{R}^{d_{\text{model}} \times d_k}$, and $W^V \in \mathbb{R}^{d_{\text{model}} \times d_v}$ are the parameter matrices.

Transformer models employ multi-head attention, which runs several attention operations in parallel. Each head uses its own set of projection matrices, allowing the model to jointly attend to information from different representation subspaces. The final output is:

\begin{equation*}
\begin{aligned}
\text{MultiHeadAttn}(X) &= \text{Concat}(\text{head}_1, \ldots, \text{head}_h)W^O \\
\text{head}_i &= \text{Attention}(XW^Q_i, XW^K_i, XW^V_i)
\end{aligned}
\end{equation*}

and $W^O \in \mathbb{R}^{hd_v \times d_{\text{model}}}$ is the output projection matrix.

\subsection{Hidden State Size}
\label{app:hidden-state}

Names of language models sometimes are demarcated with their size, such as the widely deployed Llama-3-8B, where "8B" denotes that the network contains 8 billion parameters that are used during a forward execution.  Using a 16-bit floating-point precision (2 bytes per element), this sums up to approximately 16 GB of memory that the model occupies during execution.  Due to their feed-forward nature, not all of the parameters are used at a given time—in fact, the parameters are usually layer-stratified.  Computations in neural networks take place with an operation (such as self-attention) performed on the input and the parameter, the output sometimes also called activation, or hidden state.  The size of the hidden state could be much smaller than the parameters involved, but still substantial if they were to travel on the wire if computations of different layers of the network take place on different compute services.

Specifically, with a hidden dimension of 4096 and a sequence length of 4096 tokens, each layer activation occupies:

\begin{equation*}
\text{bsize} \times \text{seq\_len} \times \text{hidden\_dim} \times \text{bytes\_per\_element}
\end{equation*}

For a batch size of 1 with sequence length of 4096 tokens (a medium-length article) using 16-bit floating-point precision, this amounts to approximately 32 MB per layer activation.  With Llama-3-8B's 32 layers, the total activation memory across all layers is around 1024 MB, or 1 GB, for a single forward pass.

At network speeds of 10 Gbps, transferring a single layer's activations would take approximately 25 milliseconds, while a full forward pass would require around 800 milliseconds just for data transfer.

\subsection{Low-Rank Adaptations (LoRA)}

Low-Rank Adaptations (LoRA) modify an existing neural network by introducing trainable low-rank matrices \cite{hu2022lora}. Specifically, given a weight matrix $W \in \mathbb{R}^{d \times k}$ in the original model, LoRA introduces the update:

\begin{equation}
\hat{W} = W + \Delta W = W + AB
\end{equation}

where $A \in \mathbb{R}^{d \times r}$ and $B \in \mathbb{R}^{r \times k}$ are low-rank matrices with rank $r \ll \min(d, k)$. During inference, the computation becomes:

\begin{equation*}
\hat{h} = \hat{W}x = Wx + ABx
\end{equation*}

This formulation allows us to maintain the original network's parameters while introducing trainable components that can modify the network's behavior with minimal parameter overhead.

LoRA is particularly effective when applied to the key matrix operations in self-attention layers, as these operations dominate the computational graph of modern language models. By applying LoRA to the $W^Q$, $W^K$, $W^V$, and $W^O$ matrices in the attention mechanism, we can efficiently modify the model's behavior while preserving its overall structure and performance characteristics.

\subsection{Language Models}

Large language models function as autoregressive predictors, generating text by sequentially predicting one token at a time based on previous tokens. This process involves three key components: tokenization, autoregressive prediction, and sampling.

\subsubsection{Tokenization}

Before processing by a neural network, text must be converted into a numerical format through tokenization:

\begin{equation*}
\text{Tokenize}: \text{String} \rightarrow \mathcal{V}^*
\end{equation*}

where $\mathcal{V}$ is a fixed vocabulary. Modern tokenizers typically employ subword techniques such as Byte-Pair Encoding (BPE) \cite{sennrich2016subword} or SentencePiece \cite{kudo2018sentencepiece}, which decompose text into frequent subword units. A text sequence $s$ is converted into tokens:

\begin{equation*}
\mathbf{x} = \text{Tokenize}(s) = (x_1, x_2, \ldots, x_M)
\end{equation*}

where each $x_i \in \mathcal{V}$ is a token from the vocabulary.

\subsubsection{Autoregressive Prediction}

Given a sequence of tokens $\mathbf{x} = (x_1, x_2, \ldots, x_T)$ where each $x_t \in \mathcal{V}$, an autoregressive model computes the conditional probability of the next token:

\begin{equation*}
\Pr(x_{t+1} | x_1, x_2, \ldots, x_t) = \mathcal{F}_{\theta}(x_1, x_2, \ldots, x_t)
\end{equation*}

For text generation, the process begins with an initial prompt $\mathbf{p} = (p_1, p_2, \ldots, p_M)$ and generates subsequent tokens through repeated application of the feedforward network $\mathcal{F}_{\theta}$:

\begin{align*}
\Pr(x_{M+1} | \mathbf{p}) &= \mathcal{F}_{\theta}(\mathbf{p}) \\
\Pr(x_{M+2} | \mathbf{p}, x_{M+1}) &= \mathcal{F}_{\theta}(\mathbf{p}, x_{M+1}) \\
&\cdots
\end{align*}

Each token prediction constitutes a complete feedforward pass through the network, creating a chain of sequentially dependent computations.

\subsubsection{Sampling}

After computing the probability distribution over the next token, a sampling strategy selects the token to generate:

\begin{equation*}
x_{t+1} = \text{Sample}(\Pr(x_{t+1} | x_1, x_2, \ldots, x_t))
\end{equation*}

Common sampling methods include:

\begin{itemize}[leftmargin=*]
    \item Greedy decoding: $x_{t+1} = \arg\max_v \Pr(v | x_1, x_2, \ldots, x_t)$
    \item Temperature sampling: $P_{\tau}(v) \propto \Pr(v | x_1, x_2, \ldots, x_t)^{1/\tau}$
    \item Top-$k$ sampling: Restricting to the $k$ most probable next tokens
    \item Nucleus (top-$p$) sampling: Restricting to the smallest set of tokens whose cumulative probability exceeds threshold $p$
\end{itemize}

We define the complete process as 
$\mathbf{dec} = \text{Decode}(\mathcal{F}_{\theta}, \mathbf{p}, T)$
where $(g_1, g_2, \ldots, g_T)$ is the generated sequence of length $T$. This process is also known as decoding.

\subsection{Batched Matrix Multiplication and Layer Parallelism}

Neural network computations are naturally amenable to parallelization through batching, which processes multiple inputs simultaneously \cite{goyal2017accurate}. In the context of service nodes, where different requests may arrive at the same node performing computations for different users, we leverage batched matrix multiplication to maintain efficiency.

For a service node processing $L$ different inputs simultaneously, each with potentially different weights, we can employ strided-batched matrix multiplication. Given inputs from $L$ different users:

\begin{equation*}
\{\mathbf{H}^{(1)}, \mathbf{H}^{(2)}, \ldots, \mathbf{H}^{(L)}\}
\end{equation*}

where each $\mathbf{H}^{(j)} \in \mathbb{R}^{S \times D}$ represents hidden states with sequence length $S$ and model dimension $D$ for user $j$, we can stack these into a 3D tensor:

\begin{equation*}
\mathcal{H} \in \mathbb{R}^{L \times S \times D}
\end{equation*}

Similarly, for the weight matrices of each layer, such as the query, key, and value projections in self-attention:

\begin{equation*}
\mathcal{W}_Q, \mathcal{W}_K, \mathcal{W}_V \in \mathbb{R}^{L \times D \times d}
\end{equation*}

where $d$ is the dimension of the attention heads.

Using batched matrix multiplication, we compute all user projections simultaneously:

\begin{equation*}
\mathcal{Q} = \text{BMM}(\mathcal{H}, \mathcal{W}_Q)
\end{equation*}

Batched computation provides substantial throughput benefits due to hardware architecture characteristics. Modern GPUs and TPUs follow a roofline model where performance is limited by either compute capability or memory bandwidth \cite{williams2009roofline}. Matrix multiplication has high arithmetic intensity, performing roughly $O(N^3)$ operations on $O(N^2)$ data. By batching operations, we increase arithmetic intensity, reduce kernel launch overhead, and maximize data reuse from fast on-chip memory. For large language models with hidden dimensions of 4096-8192, batching is essential to approach peak hardware utilization.

This batched approach to layer parallelism allows services to maintain high throughput despite the higher latency introduced by distributed computation across multiple nodes.

\end{document}